\documentclass[letterpaper, 10 pt, conference]{ieeeconf}  

\IEEEoverridecommandlockouts                              

\usepackage{cite}
\usepackage{amsmath,amssymb,amsfonts}
\usepackage{algorithm}
\usepackage{algorithmic}
\usepackage{graphicx}
\usepackage{textcomp}
\usepackage{bbm}
\usepackage{color}
\usepackage{lipsum}

\newtheorem{theorem}{Theorem}
\newtheorem{proposition}{Proposition}

\newtheorem{lemma}{Lemma}

\newtheorem{assumption}{Assumption}
\usepackage{stmaryrd}

\newcommand{\interval}[1]{\left[#1\right]}
\newcommand\intsimple[1]{\left\llbracket #1 \right\rrbracket}
\newcommand{\myvec}[2]{%
\left[
\begin{array}{c}
#1\\
#2
\end{array}
\right]
}

\newcommand{\II}{\mathcal{I}}
\newcommand{\MM}{\mathcal{M}}
\newcommand{\RR}{\mathbb{R}}
\newcommand{\TT}{\mathbb{T}}
\newcommand{\XX}{\mathcal{X}}
\newcommand{\UU}{\mathcal{U}}
\newcommand{\WW}{\mathcal{W}}
\renewcommand{\SS}{\mathcal{S}}

\newcommand{\BB}{\mathcal{B}}
\newcommand{\Ac}{\hat{A}}
\newcommand{\Bc}{\hat{B}}
\newcommand{\AKc}{\hat{A}_K}
\newcommand{\AK}{\mathcal{A}_K}
\newcommand{\MD}{\II_S}
\newcommand{\MDK}{\II_{A_K}}
\newcommand{\Mdelta}{\II_{\Delta}}
\newcommand{\bbox}[1]{\mathbb{B}\left(#1\right)}
\newcommand{\qed}{\hfill$\square$}
\newcommand{\EE}{\mathcal{E}}
\newcommand{\1}{\mathbbm{1}}

\title{\LARGE \bf Data-Driven Robust Predictive Control\\ with Interval Matrix Uncertainty Propagation}

\author{Renato Quartullo$^1$, Andrea Garulli$^2$, Mirko Leomanni$^3$\thanks{$^1$Uninettuno University, Rome, Italy. Email: renato.quartullo@uninettunouniversity.net, $^2$Universit\`a di Siena, Italy. Email: garulli@diism.unisi.it, $^3$Mercatorum University, Rome, Italy. Email: mirko.leomanni@unimercatorum.it}}

\begin{document}

	\maketitle
	\thispagestyle{empty}
	\pagestyle{empty}
	
\begin{abstract}
This paper presents a new data-driven robust predictive control law, for linear systems affected by unknown-but-bounded process disturbances. 
A sequence of input-state data is used to construct a suitable uncertainty representation based on interval matrices. Then, the effect of uncertainty along the prediction horizon is bounded through an operator leveraging matrix zonotopes. This yields a tube that is exploited within a variable-horizon optimal control problem, to guarantee robust satisfaction of state and input constraints. The resulting data-driven predictive control scheme is proven to be recursively feasible and practically stable. 
A numerical example shows that the proposed approach compares favorably to existing methods based on zonotopic tubes.
\end{abstract}

\section{Introduction}\label{sec:introduction}

Data-driven control has been an active research area for long time, receiving increasing attention in recent years due to the striking development of machine learning techniques \cite{dorfler2023data,soudbakhsh2023data}.
Within this broad research landscape, Data-Driven Predictive Control (DDPC) has emerged as a promising control design framework, after the publication of the seminal work \cite{coulson2019data}. In that paper, an approach based on behavioural system theory is adopted to directly predict future trajectories from a set of past input-output data. 
This allows one to formulate a predictive control scheme without explicitly defining the system dynamics, which would require either to resort to first-principles modeling or to apply some system identification procedure.
When data are not affected by noise, the solution proposed in \cite{coulson2019data} turns out to be equivalent to classical model-based predictive control (MPC), thus inheriting its well-known capability to handle multivariable constrained control problems  \cite{rawlings2017model,borrelli2017predictive}. Moreover, in \cite{fiedler2021relationship} it is shown that DDPC is also equivalent to subspace predictive control \cite{favoreel1999spc}, an indirect approach in which the prediction model is identified from data.
An intense research activity has been devoted to the setting in which data are affected by noise. The case of bounded output noise is treated in \cite{berberich2021data}, by using slack variables and a regularization term. Regularization is also used to deal with stochastic output noise in \cite{fiedler2021relationship,dorfler2022bridging}. A unified framework encompassing several regularization-based approaches is presented in \cite{breschi2023data}.
Closed-loop stability properties have been also investigated. A terminal equality constraint is used in \cite{berberich2021data} to derive conditions for recursive feasibility and practical exponential stability of a DDPC scheme with bounded output noise. A similar approach is proposed in \cite{kloppelt2025novel} for systems affected by bounded process disturbances. The key idea is to cope with the effect of disturbances by suitably tightening the state and input constraints, similarly to what is done in tube-based MPC \cite{chisci2001systems,mayne2005robust}.  
For a review of the theoretical properties and implementation details of different versions of DDPC, the interested reader is referred to \cite{verheijen2023handbook,berberich2025overview}.

Recent works have addressed DDPC for systems affected by bounded process disturbances, exploiting results on reachability analysis based on matrix zonotopes \cite{alanwar2021data}. In \cite{alanwar2022robust}, a robust control scheme is presented in which data-driven reachable regions are used to achieve robust satisfaction of the state and input constraints along the prediction horizon. A similar approach is taken in \cite{russo2023tube}, where probabilistic stability of the resulting scheme is shown. Both methods do not address recursive feasibility, which is instead guaranteed by the solution proposed in \cite{farjadnia2024robust}. However, this result is based on a rigid tube, whose computation requires the knowledge of the data covering radius (namely, a measure of the data density in the input-state space), which is not available a priori and whose estimate form data is not easy \cite{alanwar2023data}. 

In this paper, a new data-driven robust predictive control scheme is proposed, for linear discrete-time systems affected by unknown-but-bounded process disturbances. The set of systems compatible with a collection of state-input data is characterized in terms of interval matrices, in two distinct ways: using a data-driven approach and via set-membership identification. Then, the uncertainty propagation along the prediction horizon is approximated by using a suitable operator based on matrix zonotopes.
The considered interval matrix uncertainty representation has a twofold benefit: i) it allows one to guarantee recursive feasibility and practical stability of a variable-horizon robust predictive control scheme, without requiring any assumption on the data density;  ii) it provides an efficient way to compute a tube bounding the uncertainty propagation, which turns out to be less conservative than that used in~\cite{farjadnia2024robust}.
A further contribution is the comparison of the two proposed robust predictive control laws, based respectively on data-driven and set-membership uncertainty representation, showing that the two approaches yield comparable results for process disturbances of moderate size.

The rest of the paper is organized as follows. Section \ref{sec:notation} introduces some preliminary material on matrix operations and set approximations. The problem formulation is presented in Section \ref{sec:problem}. The proposed solution based on interval matrix uncertainty propagation is detailed in Section \ref{sec:intervalMPC}, which also states the closed-loop properties of the control scheme. Section \ref{sec:numerical} reports numerical results showing the performance of the proposed approach. Finally, some conclusions are drawn in Section \ref{sec:conclusions}.

	\section{Notation and preliminaries}
    \label{sec:notation}

For sets of real matrices $\mathcal{A}$, $\mathcal{B}$, the following operations are considered:
$\mathcal{A}\oplus\mathcal{B} = \{A+B|\,\,A\in\mathcal{A},B\in\mathcal{B}\}$,
$\mathcal{A}\ominus\mathcal{B} = \{M|\,\,M \oplus \mathcal{B}  \subseteq \mathcal{A}\}$,
$\mathcal{A}\mathcal{B} = \{AB|\,\,A\in\mathcal{A},B\in\mathcal{B}\}$, $\mathcal{A}^j = \{A^j|\,\,A\in\mathcal{A}\}$.
An interval matrix $\II=\{M | \,\, C-\Delta \leq M \leq C+\Delta \}$, where
the operator $\leq$ represents elementwise inequality, is compactly denoted as $\II = C \oplus 
\intsimple{\Delta}$.
A matrix zonotope
$
\MM = \Big\{ M | 
M = M_C + \sum_{i=1}^{n_g} G_i \beta_i,\; |\beta_i|\leq 1 \Big\}$,
where $M_C$ is the center and 
$G_i$, $i=1,\dots,n_g$, 
are the generators, is denoted as $\MM = \langle M_C;\,G_1, \ldots,\, G_{n_g} \rangle$.

The minimum interval matrix containing a matrix zonotope 
$\MM = \langle M;\, G_1, \ldots, G_{n_g} \rangle$,  is given by
\begin{equation}
\bbox{\MM} = M \oplus \intsimple{\Delta},
\label{eq:boxoutzon}
\end{equation} 
where $\Delta = \sum_{i=1}^{n_g} |G_i|.$

Given a matrix $ M \in \mathbb{R}^{l \times q} $, the set of matrices $\mathcal{E}(M) = \left\{ M^{(k)} \right\}_{k=1}^{lq}$, contains all matrices $ M^{(k)} \in \mathbb{R}^{l \times q} $ defined as
$$
M^{(k)}_{ij} =
\begin{cases}
M_{ij}, & \text{if } k = (i - 1)q + j, \\
0, & \text{otherwise},
\end{cases}
$$
so that $\sum_{k=1}^{lq} M^{(k)} = M$.

The following result provides a key set approximation tool that will be exploited for uncertainty propagation.
\begin{proposition} \label{prop:overapprox}
Let $\II = C \oplus \intsimple{\Delta} \in \mathbb{R}^{l \times p}$ be an interval matrix and $\MM=\langle M;\, G_1, \ldots, G_{n_g}\rangle\subseteq \mathbb{R}^{p \times q}$ be a matrix zonotope. Then 
\begin{equation}\label{eq:Topdef}
\II\MM \subseteq \TT_{\II}(\MM) = \langle CM;\, CG_1, \ldots, CG_{n_g},\, F^{(1)}, \ldots, F^{(lq)} \rangle,
\end{equation}
where 
\begin{equation} \label{eq:Fgen}
F \;=\; {\Delta}\!\left( \,\lvert M \rvert + \sum_{j=1}^{n_g} \lvert G_j \rvert \,\right)
\end{equation}
and $\{F^{(i)}\}_{i=1}^{lq} = \mathcal{E}(F)$.
\end{proposition}
\emph{Proof.} See Appendix.

If $j$ successive iterations of the operator $\TT_{\II}$ are performed on the matrix zonotope $\MM$, one obtains the set approximation
\begin{equation}
\label{eq:IjMapprox}
\II^j\MM\subseteq\TT_{\II}^j(\MM)   =  \underbrace{\TT_{\II} \circ \TT_{\II} \circ \cdots \circ \TT_{\II}}_{j \text{ times}} (\MM) .
\end{equation}

       \section{Problem Formulation}\label{sec:problem}
        Consider the linear time-invariant system
        \begin{equation}\label{eq:sys}
            x(k+1) = Ax(k) + Bu(k) + w(k),
        \end{equation}
        where $x(k) \in\RR^n$ is the system state at time $k$, $u(k) \in\RR^m$ is the control input, $w(k)\in\RR^n$ is the process disturbance and $A,\, B$ are unknown system matrices. 
        It is assumed that $w(k)$ is bounded componentwise, i.e. 
        $$
        w(k) \in\WW= \intsimple{\overline{w}},
        $$
        with $\overline{w} \in \RR^n$, $\overline{w} \geq 0$.

        Suppose that a sequence of input-state data of length $T$ generated by system \eqref{eq:sys} is available. These samples are organized into the data matrices
        \begin{align}
        X_+ &= \left[x(1)\; x(2)\; \cdots\; x(T)\right], \label{eq:Xplus}\\
        X_- &= \left[x(0)\; x(1)\; \cdots\; x(T-1)\right],\label{eq:Xminus}\\
        U_- &= \left[u(0)\; u(1)\; \cdots\; u(T-1)\right].\label{eq:Uminus}    
        \end{align}
        
       The objective of this work is to design a robust predictive control scheme using only the available input-state data, without explicit knowledge of the system matrices $A$ and $B$. The resulting closed-loop system must satisfy the state and input constraints
        \begin{equation}\label{eq:constraint}
            x(k) \in \XX,\,u(k) \in \UU,
        \end{equation}
        for any possible realization of $w(k)$,
        with $\XX$, $\UU$ being convex polytopic sets. 
        
        To proceed, we first describe two different ways to compute interval matrix bounds for the unknown system matrices $A$ and $B$, that are consistent with the collected dataset: one directly based on data matrices \eqref{eq:Xplus}-\eqref{eq:Uminus}, and one exploiting set-membership identification. Then, a general framework is introduced as a basis for the proposed data-driven robust predictive control approach. 
        
        \subsection{Data-driven Interval System Matrices}\label{subsec:DD}
            Let us define the state-input data matrix
            \begin{equation}
            \Phi =
            \begin{bmatrix}
            X_- \\
            U_-
            \end{bmatrix}
            \in \mathbb{R}^{(n+m)\times T}.
            \label{eq:datamat}
            \end{equation}
            Then the trajectories collected from system~\eqref{eq:sys} satisfy
            \begin{equation}\label{eq:data_relation}
            X_+ = \Theta \Phi + W,
            \end{equation}
            where
            $$
            \Theta = \left[A \quad B\right] \in \mathbb{R}^{n \times (n+m)},
            $$
            and $W = \left[w(0)\; w(1)\; \cdots \; w(T-1) \right]$. 
            The following assumption is made, which is standard in data-driven control
            \begin{equation}
            \operatorname{rank}(\Phi) = n+m.
            \label{eq:rankass}
            \end{equation}
	   Now, let $\Sigma_\Theta$ be the set of all matrices $\Theta$ that are compatible with the given data, i.e.
	   $$
	    \Sigma_\Theta=\left\{ \Theta: ~ X_+ - \Theta \Phi \in \intsimple{\,\overline{W}\,} \right\}
	    $$
	    where $\overline{W}=[ \overline{w} ~ \overline{w} \dots \overline{w}]$.
	    The following result holds.
            \begin{proposition}
            \label{prop:intsys}
            Let \eqref{eq:rankass} hold. Then
            \begin{equation}
            \Sigma_\Theta \subseteq X_+ \Phi^\dag \oplus \intsimple{\Delta_S}
            \label{eq:intsys}
            \end{equation}
            with $\Delta_S\!=\!\overline{w} \, \1^\top | \Phi^\dag |$,
            where $\1$ denotes the all-ones vector and $\Phi^\dag$ is the right pseudoinverse of $\Phi$.
            \end{proposition}
            \emph{Proof.} 
            Consider the matrix zonotope
            \begin{equation}
            \MM_\Theta= \left( X_+ \oplus \intsimple{\,\overline{W}\,} \right) \,  \Phi^\dag. 
            \label{eq:MTheta}
            \end{equation}
            From Lemma 1 in \cite{alanwar2021data}, one has $\Sigma_\Theta \subseteq \MM_\Theta$.
            Observe that the generators of $\intsimple{\,\overline{W}\,}$ are the matrices $\overline{w}_i e_i e_j$, for $i=1,\dots,n$, $j=1,\dots,T$, where $e_i$ denotes the $i$-th canonical basis vector. Then, the zonotope \eqref{eq:MTheta} can be 	written as
            $$
            \MM_\Theta= \langle  X_+ \Phi^\dag ; G_1, \dots , G_{nT} \rangle
            $$
            with $G_h=\overline{w}_i e_i e_j^\top \Phi^\dag$, for $h=i+(j-1)n$. 
            Hence, \eqref{eq:intsys} follows by computing $\bbox{\MM_\Theta}$ according to \eqref{eq:boxoutzon}, with
            $$
            \begin{array}{rcl}
            \Delta_S &=& \displaystyle{ \sum_{i=1}^n \sum_{j=1}^T | G_{i+(j-1)n} | = \sum_{i=1}^n \overline{w}_i e_i \sum_{j=1}^T e_j^\top |\Phi^\dag | } \vspace*{2mm}\\ 
            &=&  \overline{w} \, \1^\top | \Phi^\dag |.
            \vspace*{-4mm}
            \end{array}
            $$           
	    \qed	
	    
	    In the remainder of this work we will consider the following partition of the data-based interval system matrices in \eqref{eq:intsys}
	    \begin{align}
	    X_+ \Phi^\dag &= \interval{\Ac\quad\Bc}  , \label{eq:ddmodel1}\\
	    \overline{w} \, \1^\top | \Phi^\dag | &= \interval {\Delta_A\quad \Delta_B}.
        \label{eq:ddmodel2}
	    \end{align}
	    Due to Proposition \ref{prop:intsys}, the unknown system matrices in \eqref{eq:sys} satisfy
	    \begin{equation}\label{eq:sys2}
                \interval{A\quad B}\in\MD =\interval{\Ac\quad\Bc}\oplus \intsimple{\Delta_A\quad \Delta_B}.
            \end{equation}
            Therefore, the properties of any robust predictive control scheme designed on the uncertain system 
            \begin{equation}\label{eq:uncsys}
            x(k+1) = \Ac x(k) + \Bc u(k) + A_{\Delta} x(k) + B_{\Delta} u(k) +w(k),
            \end{equation}
            where $A_{\Delta} = A-\Ac$, $B_{\Delta} = B-\Bc$ and $\left[A_{\Delta}\quad B_{\Delta}\right]\in \intsimple{\Delta_A\quad \Delta_B}=\intsimple{\Delta_S}$, are guaranteed to hold also for the original system \eqref{eq:sys}.

            \subsection{Set-membership identification of Interval System Matrices}\label{subsec:SM}
            
            An alternative to the data-driven uncertain system representation previously introduced, is provided by set-membership system identification~\cite{milanese1991optimal}.
            Since $w(k) \in \intsimple{\overline{w}}$, from \eqref{eq:sys} it follows that
            \begin{equation}
            x(k+1) - \overline{w}
            \;\le\;
            \Theta \phi(k)
            \;\le\;
            x(k+1) + \overline{w},
            \label{eq:data_ineq}
            \end{equation}
            where $\phi(k) = \left[ x(k)^\top \quad u(k)^\top\right]^\top$.
            Note that the constraints in \eqref{eq:data_ineq} are separable across the rows of matrix $\Theta$. Let $\theta_i^T$ denote the $i$-th row of $\Theta$.
            Then, for each state component $x_i$, $i = 1,\dots,n$, one gets
            \begin{equation}
            x_i(k+1) - \overline{w}_i \leq \theta_i^\top\phi(k) \leq
            x_i(k+1) + \overline{w}_i,
            \label{eq:data_ineq_row}
            \end{equation}
            Stacking the inequalities over all time steps $k$ yields the equivalent compact representation
            \begin{equation}\label{eq:data_ineq_row_stack}
                X_{+,i} - \overline{w}_i \1^\top \leq\theta_i^\top\Phi\leq X_{+,i}+\overline{w}_i \1^\top ,
            \end{equation}
             where $X_{+,i}$ is the $i$-th row of $X_+$. 
	   Therefore, the feasible set for each $\theta_i$ is a bounded polyhedron, and the tightest interval bounds on the entries of $\Theta$ can be computed element-wise by solving  $n(n+m)$ linear programs subject to the constraints \eqref{eq:data_ineq_row_stack}. 
       Once such bounds are derived, an uncertainty model in the form \eqref{eq:sys2} is available 
       and a robust predictive control scheme can be designed on system \eqref{eq:uncsys}.

        \subsection{Robust Predictive Control}
            Let us define the nominal system 
            \begin{equation}\label{eq:sys_nominal}
                z(k+1) = \Ac z(k) + \Bc v(k)
            \end{equation}
            and introduce the feedback policy            
            \begin{equation}\label{eq:u_RMPC}
                u(k) = K(x(k)-z(k)) + v(k),
             \end{equation} 
            where $v(k)$ is a nominal control input and $K\in\RR^{m \times n}$ is a gain matrix.
            By using \eqref{eq:sys2}, matrix $A_K=A+BK$ satisfies 
            \begin{equation}
            \label{eq:cAK}
                A_K\in\AK = \AKc \oplus \Mdelta\myvec{I}{K},
            \end{equation}
            where $\AKc = \Ac+\Bc K$ and $\Mdelta=\intsimple{\Delta_S}$.
            
            Let $e(k) = x(k)-z(k)$ be the error between the true and the nominal state. From \eqref{eq:sys} and \eqref{eq:uncsys}, the dynamics of $e(k)$ is given by
            \begin{equation}\label{eq:error_dyn2}
                e(k+1) = A_K e(k) + A_{\Delta} z(k) + B_{\Delta} v(k) + w(k).
            \end{equation}
            Then, from \eqref{eq:cAK}-\eqref{eq:error_dyn2}, the set-valued dynamics
            \begin{equation}\label{eq:Sk_recursion}
                \SS(k+1) = \AK\SS(k) \oplus \Mdelta\left[\begin{array}{c}
                     z(k)\\
                     v(k)
                \end{array}\right] \oplus \WW
            \end{equation}
            is such that if $e(0)\in\SS(0)$, then $e(k)\in\SS(k)$, $\forall k>0$.
            
            With the aim to design a robust predictive control scheme, let us introduce the following variable-horizon robust control problem
            \begin{equation}\label{eq:mpc_tube_exact} 
                \begin{aligned}
                    \underset{N_k,\mathbf{v}(k),\mathbf{z}(k)}{\text{min }} &   J_k(N_k,\mathbf{v}(k),\mathbf{z}(k))\\
                    \text{s.t.} \quad & z_k(0)= x(k)\\
                    & {z_k}(j+1)=\Ac {z_k}(j)+\Bc {v_k}(j)  \\
                    &  {z_k}(j)\oplus\SS_k(j) \subseteq \XX, \quad j=0,\ldots,N_k \\
                    &  v_k(j)\oplus K\SS_k(j)\subseteq \UU, \quad j=0,\ldots, N_k-1 \\
                    &  z_k(N_k) \oplus \SS_k(N_k) \subseteq\XX_{f} \\
                    & N_k \in \mathbb{N}^+
                \end{aligned}
            \end{equation}
            in which $\mathbf{v}_k= \left(v_k(0),\ldots,\,v_k(N_k-1)\right)$ and $\mathbf{z}_k = \left(z_k(0),\ldots,\,z_k(N_k)\right)$ are the sequences of nominal inputs and states, respectively, and 
            the cost $J_k$ is defined as 
            \begin{equation}\label{eq:cost}
                J_k = \gamma N_k + \sum_{j=0}^{N_k-1} \ell(v_k(j)-Kz_k(j))
            \end{equation}
            where $\ell(\cdot)$ is such that $\ell(x)\geq 0,\,\forall x$, and $\gamma> 0$ is a tuning parameter weighting the length of the prediction horizon.
            The choice of including the prediction horizon length $N_k$ among the optimization variables is instrumental to ensure recursive feasibility of the resulting predictive control scheme, as it will be clarified later.
           Similarly to \eqref{eq:Sk_recursion}, the sets $\SS_k(j)$ in \eqref{eq:mpc_tube_exact} are defined by the set-valued dynamics
            \begin{equation}\label{eq:setdyn}
            \SS_k(j+1) = \AK\SS_k(j)\oplus\Mdelta\left[\begin{array}{c}
                 z_k(j)\\
                 v_k(j) 
            \end{array}\right]\oplus\WW
            \end{equation}
with  $\SS_k(0) = \{0\}$, due to the initial constraint $z_k(0) = x(k)$. 
            The terminal set $\mathcal{X}_f$ is a further design element that can be used to guarantee stability properties of the control scheme.  
            
            As in classical tube-based MPC, the constraints of problem \eqref{eq:mpc_tube_exact} coupled with the control law \eqref{eq:u_RMPC}, guarantee that the state and input constraints \eqref{eq:constraint} are robustly satisfied along the prediction horizon. To this aim, \eqref{eq:mpc_tube_exact}  explicitly incorporates uncertainty propagation along the prediction horizon through the sets $\SS_k(j)$, in the same spirit as in \cite{alanwar2022robust,russo2023tube}. In particular, by exploiting \eqref{eq:setdyn}, such sets can be expressed as
            \begin{equation}\label{eq:Sk_response}
            \SS_k(j) = \sum_{i=0}^{j-1}\AK^{j-i-1}\Mdelta\left[\begin{array}{c}
                     z_k(i)\\
                     v_k(i)
                \end{array}\right] \oplus \sum_{i=0}^{j-1}\AK^i\WW.
            \end{equation}
            Unfortunately, the exact computation of the sets $\SS_k(j)$ becomes intractable even for small values of $j$. Notice that a further difficulty stems from the fact that their expression in \eqref{eq:Sk_response} depends on the optimization variables $\mathbf{v}_k$ and $\mathbf{z}_k$ of problem~\eqref{eq:mpc_tube_exact}. This issue has been tackled in \cite{farjadnia2024robust} by replacing the sets $\SS_k(j)$ with a constant outer bound $\SS$, corresponding to a rigid tube. However, this choice turns out to be overly conservative and relies on the knowledge of the covering radius of the dataset. In the next section, a less conservative approximation is proposed, which allows to maintain the time-varying nature of uncertainty propagation along the prediction horizon within the optimization problem, without requiring additional assumptions on the dataset. 
            
            \section{Interval-Matrix Robust Predictive Control}
            \label{sec:intervalMPC}
            
            In order to construct an outer approximation of the sets $\SS_k(j)$, we exploit uncertainty propagation in the form of matrix zonotopes through the operator $\TT$ in \eqref{eq:Topdef}.  In particular, the operator $\TT^j$ in \eqref{eq:IjMapprox} is used to outer bound the set multiplication terms in \eqref{eq:Sk_response}. 
            Let us define the set 
            \begin{equation}\label{eq:Bbox_Delta}
                \BB_k^{\Delta}(j) = \sum_{i=0}^{j-1} \II^{\Delta}(j-i-1)\myvec{z_k(i)}{v_k(i)}
            \end{equation}
	    where
	     \begin{equation}\label{eq:Ibox}
                \II^{\Delta}(j) =  \bbox{\TT^j_{\MDK}(\II_{\Delta})},
            \end{equation}
           and
            $\MDK = \AKc \oplus \intsimple{\Delta_K}$ with $\Delta_K=\Delta_A + \Delta_B |K|$. 
           Similarly, let
	    \begin{equation}\label{eq:Bbox_W}
                \BB^{\WW}(j) = \sum_{i=0}^{j-1} \II^{\WW}(i)
            \end{equation}
            where
             \begin{equation}\label{eq:Ibox_W}
                \II^{\WW}(j) =  \bbox{\TT^j_{\MDK}(\WW)}.
            \end{equation}
            Then, define
             \begin{equation}\label{eq:Bbox}
                \BB_k(j) = \BB_k^{\Delta}(j) \oplus\BB^{\WW}(j).
            \end{equation}
            The following result holds.
             \begin{proposition}
            \label{prop:Bk}
            The sets $\BB_k(j)$ in \eqref{eq:Bbox} satisfy the inclusion
            \begin{equation}\label{eq:SinB}
            \SS_k(j) \subseteq \BB_k(j),\quad \forall k,j.
            \vspace*{2mm}
            \end{equation}
            \end{proposition}
            \emph{Proof.} 
            From \eqref{eq:cAK} 
            one has that $\AK  \subseteq \MDK$. Hence, as a consequence of \eqref{eq:IjMapprox} one gets $\AK^{j}\Mdelta \subseteq \II^{\Delta}(j)$ in \eqref{eq:Ibox} and $\AK^{j}\WW \subseteq \II^{\WW}(j)$ in~\eqref{eq:Ibox_W}. Therefore, the result follows by comparing \eqref{eq:Sk_response} and \eqref{eq:Bbox}.\qed

	   By exploiting the outer approximation of the sets $\SS_k(j)$ provided by Proposition \ref{prop:Bk},
   	   the robust optimal control problem~\eqref{eq:mpc_tube_exact} is reformulated as follows
	    \begin{subequations}\label{eq:mpc_tube_box}
                \begin{align}
                \min_{N_k,\mathbf v(k),\mathbf z(k)}\ & 
                J_k(N_k,\mathbf v(k),\mathbf z(k)) \notag\\
                \text{s.t.}\quad 
                & z_k(0)= x(k) \label{subeq:mpc_init}\\
                & z_k(j+1)=\Ac z_k(j)+\Bc v_k(j), \label{subeq:mpc_dynamic}\\
                & \begin{aligned}
                   z_k(j)\oplus\BB_k(j) &\subseteq \XX,   \quad  j=0,\dots,N_k
                  \end{aligned}
                  \label{subeq:mpc_state}\\
                & \begin{aligned}
                   v_k(j)\oplus K\BB_k(j) &\subseteq \UU,  \quad  j=0,\dots, N_k-1
                  \end{aligned}
                  \label{subeq:mpc_input}\\
                & z_k(N_k) \oplus \BB_k(N_k) \subseteq \XX_f 
                  \label{subeq:mpc_terminal}\\
                & N_k \in \mathbb{N}^+. \notag
                \end{align}
            \end{subequations}         
            Denoting by $N^*_k$, $\mathbf{v}^*(k)$, $\mathbf{z}^*(k)$ the optimal solution of problem~\eqref{eq:mpc_tube_box}, at each time $k$ the nominal control input in \eqref{eq:u_RMPC} is chosen as
            \begin{equation}\label{eq:optvsel}
                v(k)=v^*_k(0).
            \end{equation}
             like in standard MPC schemes.
            
             The only difference between problems \eqref{eq:mpc_tube_box} and \eqref{eq:mpc_tube_exact} is that the constraints \eqref{subeq:mpc_state}-\eqref{subeq:mpc_terminal} are more conservative than the corresponding constraints in problem \eqref{eq:mpc_tube_exact}, due to the inclusion \eqref{eq:SinB}. This implies that the closed-loop system resulting from the application of the control law \eqref{eq:mpc_tube_box}-\eqref{eq:optvsel} to system \eqref{eq:sys} will robustly satisfy the constraints \eqref{eq:constraint}, whatever are the system matrices $A,B$ and the process disturbance realization $w(k)\in\WW$.  
             
             It is worth stressing that, while the sets $\BB_k^{\Delta}(j)$ depend on the nominal state and input sequences $\mathbf{z}_k$ and $\mathbf{v}_k$, which are optimization variables of problem \eqref{eq:mpc_tube_box}, the interval matrices $\II^{\Delta}(j)$ and the sets  $\II^{\WW}(j)$, $\BB^{\WW}(j)$ can be computed offline for all $j$, since they are independent of $\mathbf{z}_k$ and $\mathbf{v}_k$. This is a key feature of the proposed approach because it allows to significantly limit the computational burden required for the online solution of problem \eqref{eq:mpc_tube_box}.

%

In order to derive the closed-loop properties of the control scheme \eqref{eq:mpc_tube_box}-\eqref{eq:optvsel}, the following assumptions are enforced.
            \begin{assumption}\label{assum:K}
                The matrix $K$ in \eqref{eq:u_RMPC} is such that $A_K = A+B K$ is a Schur matrix for all $\left[A\quad B\right]\in\MD$ in \eqref{eq:sys2}.
            \end{assumption}

	\begin{assumption}\label{assum:RPI}
            The terminal set $\XX_f$ in \eqref{subeq:mpc_terminal}
            is a robust positive invariant (RPI) set for the system $ x(k+1) = A_K x(k) +w(k)$, which means that $A_K x+w\in\XX_f$, $\forall x\in \XX_f$, $\forall A_K \in \AK$ in \eqref{eq:cAK}, and $\forall w\in\WW$. Moreover, $\XX_f$ is such that $\XX_f \subseteq\XX$ and  $K\XX_f \subseteq\UU$. 
            \end{assumption}
              
            \begin{assumption}\label{assum:P0}
                Problem~\eqref{eq:mpc_tube_box} is feasible at time $k=0$ with associated optimal cost $J_0^*$.
            \end{assumption}
            
            Assumptions \ref{assum:K}-\ref{assum:P0} are common in robust MPC literature. In particular, methods for computing RPI sets include iterative algorithms \cite{kolmanovsky1998theory},~\cite[Ch.~10]{borrelli2017predictive} or LMI-based approaches \cite[Ch.~5]{CANNONbook}, leading to ellipsoidal or polytopic terminal sets.

            We are now ready to state the main results of this section.
            
            \begin{theorem}\label{thm:RF}
            Let Assumptions~\ref{assum:K}-\ref{assum:P0} be satisfied. Then, problem~\eqref{eq:mpc_tube_box} admits a feasible solution for all $k>0$, for any $\left[A\quad B\right]\in\MD$ and any realization of $w(k)\in\WW$.
            \end{theorem}
            \emph{Proof.} See Appendix.
            
            \begin{theorem}\label{thm:convergence}
                Let Assumptions~\ref{assum:K}-\ref{assum:P0} be satisfied. Then, the closed-loop trajectories $x(k)$ of system~\eqref{eq:sys}, under the control law \eqref{eq:u_RMPC}, \eqref{eq:mpc_tube_box}-\eqref{eq:optvsel}, reach the set $\XX_f$ in at most $\lfloor J_0^*/\gamma\rfloor$+1 steps and converge asymptotically to a set included in $\XX_{\infty} = \sum_{i=0}^{\infty}\TT_{\MDK}^i(\WW)$.
            \end{theorem}
            \emph{Proof.} See Appendix.
            
 From a computational viewpoint, problem \eqref{eq:mpc_tube_box} can be solved very efficiently. 
Thanks to the use of interval matrices, constraints in \eqref{subeq:mpc_state}-\eqref{subeq:mpc_terminal} can be reformulated as linear constraints in the following way.
By using \eqref{eq:Bbox_Delta}, \eqref{eq:Bbox_W} and \eqref{eq:Bbox}, the $j$-th state constraint~\eqref{subeq:mpc_state} can be rewritten as
 \begin{equation}
 \label{eq:state_rw}
 z_k(j) \oplus \sum_{i=0}^{j-1}\II^{\Delta}(j-1-i)\myvec{z_k(i)}{v_k(i)}\subseteq \XX\ominus\BB^\WW(j).
 \end{equation}
Since the state constraints $\XX$ is a polytopic set, the tightened constraint set on the right hand side can be expressed as 
$$
\XX\ominus\BB^\WW(j) = \{x\in\RR^n|\,\,H(j) x\leq b(j)\}.
$$
Observing that all interval matrices $\II^{\Delta}(j)$ in~\eqref{eq:Ibox} are centered in zero, let us write them as  $\II^{\Delta}(j) = \intsimple{\Delta_I(j)}$. 
Then, \eqref{eq:state_rw} turns out to be equivalent to 
 \begin{equation}\label{eq:constr_efficient}
H(j) z_k(j) + \left|H(j) \right|\sum_{i=0}^{j-1} \Delta_I(j-i-1)\left|\myvec{z_k(i)}{v_k(i)}\right|\leq b(j)
\end{equation}
which is easily cast as a linear constraint in the variables $\mathbf{z}_k$, $\mathbf{v}_k$
\cite{boyd2004convex}.
The input constraint \eqref{subeq:mpc_input} and the terminal constraint  \eqref{subeq:mpc_terminal} can be handled in the same way.

It is worth remarking that the advantage of the proposed approach stems from the combined use of two distinct features. First, the zonotopic uncertainty propagation along the prediction horizon, through the operator $\TT^j$, permits to explicitly incorporate the effect of the data-driven uncertainty model \eqref{eq:ddmodel1}-\eqref{eq:ddmodel2} in the robust optimal control problem \eqref{eq:mpc_tube_box}.
Secondly, the adoption of the interval-matrix-based tube $\BB_k(j)$ has a twofold effect: (i) it is key to prove recursive feasibility of the predictive control scheme; (ii) it allows one to cast the constraints of problem \eqref{eq:mpc_tube_box} in the computationally lightweight form \eqref{eq:constr_efficient}.

The adoption of a variable-horizon predictive control scheme is instrumental to achieve recursive feasibility. 
This makes problem~\eqref{eq:mpc_tube_box} a mixed-integer program with linear constraints and only one integer variable $N_k$. 
Such a problem can be solved exactly by using standard solvers (see, e.g., \cite{gurobi}), or by adopting efficient heuristic procedures~\cite{leomanni2022variable,persson2024optimization}.

        \section{Numerical results}\label{sec:numerical}
            
            In this section, numerical results for a double-integrator example are presented to evaluate the performance of two control schemes: i) the Data-Driven Interval Matrix Predictive Control (DD-IMPC) scheme, in which the control law \eqref{eq:u_RMPC}, \eqref{eq:mpc_tube_box}-\eqref{eq:optvsel} is applied to the data-driven interval matrices \eqref{eq:intsys}; ii) the Set-Membership Interval Matrix Predictive Control (SM-IMPC) scheme, in which the same control law is applied to the uncertainty model~\eqref{eq:sys2}, estimated via the set-membership procedure in Section~\ref{subsec:SM}. These control strategies are compared to the Tube-Based Zonotopic Predictive Control (TZPC) method introduced in~\cite{farjadnia2024robust}. 
            The simulations are carried out using MATLAB on a laptop equipped with the Apple M4 processor and 16~GB of RAM. The code is written in MATLAB and is publicly available at {https://github.com/RenatoQuartullo/Data-Driven-Interval-Matrices-PC}.

            Consider the discrete-time system~\eqref{eq:sys},
            in which the true system matrices are 
            	\begin{equation*}
            		A=
            		\left[
            		\begin{array}{c c}
            			1&\quad 1\\
            			0&\quad 1 \end{array}\right],
            		\quad
                    B=\left[
            		\begin{array}{lll}
            			0\\
            			1\end{array}\right]
            	\end{equation*}
            and $\WW = \intsimple{\overline{w}}$ with $\overline{w}= [0.1,\,0.05]^\top$.
            The state constraint $\mathcal{X}$ is such that $$\myvec{-12}{-4}\leq x(k) \leq\myvec{12}{4},$$ while the input command must satisfy $|u(k)|\leq 2$. The matrix gain $K$ and the terminal set $\XX_f$ are calculated following the procedures described in~\cite[Sec.~5.2]{CANNONbook} and~\cite[Ch.~10]{borrelli2017predictive}, respectively. In DD-IMPC and SM-IMPC the function $\ell(\cdot)$ in the cost $J_k$ is quadratic, i.e., $J_k = \gamma N_k+\sum_{j=0}^{N_k-1} \|v_k(j)-Kz_k(j)\|_2^2$ with $\gamma = 1$. The variable-horizon problem~\eqref{eq:mpc_tube_box} is solved by enumeration for each horizon length $N_k \in\{ 1,\ldots,10\}$.
            The true system is used only to generate the data matrices $X_-,\,X_+,\,U_-$, collecting a trajectory of length $T = 50$.

            First, we evaluate the feasible domain (FD) of the considered control schemes and its sensitivity to the available data. The FDs are computed as the convex hull of the initial conditions for which the initial optimization problem is feasible. A total of 300 initial conditions are generated by uniformly gridding the admissible state set~$\XX$. 
            To ensure a fair comparison, the feasible domain of TZPC (which relies on a fixed prediction horizon) is defined as the set of all initial states for which the corresponding optimization problem is feasible for at least one horizon length $N \in \{1,\ldots,10\}$. The FD evaluation is repeated 50 times using independently generated input-state datasets.
            Table~\ref{tab:FD} reports for each method the average area of the FDs, its standard deviation and the number of datasets for which the FD is empty. The results show that DD-IMPC achieves a larger FD on average and exhibits significantly lower sensitivity to the available data compared to TZPC. 
            It is also worth noting that the FD of DD-IMPC is very close to that yielded by SM-IMPC (which is essentially unaffected by the dataset), although the latter exploits an intermediate identification step.
            Figure~\ref{fig:FD} illustrates the FDs obtained in one representative run in which TZPC produces a non-empty feasible domain.
            \begin{table}[t]
                \centering
                \caption{Statistics of the feasible domain.}
                \begin{tabular}{lccc}
                \hline
                Method & Average & Standard Deviation & Empty FD \\
                \hline
                DD-IMPC & 147.2 & 13.4 & 0 \\
                TZPC & 98.5 & 62.2 & 14 \\
                SM-IMPC & 162.8 & 0 & 0 \\
                \hline
                \end{tabular}
                \label{tab:FD}
            \end{table}
            \begin{figure}[h]
                \centering
                \includegraphics[width=\columnwidth]{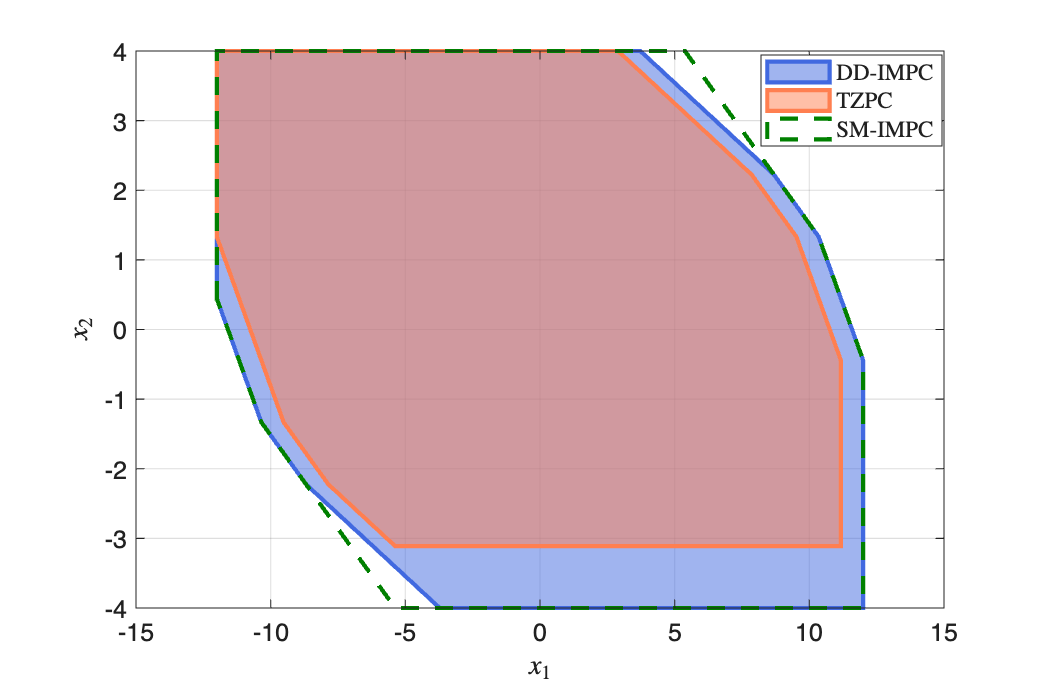}
                \caption{Feasible domain comparison.}
                \label{fig:FD}
            \end{figure}

            To assess the sensitivity of the considered control schemes to the size of the process disturbance, we consider $\WW = \epsilon_W \intsimple{\overline{w}}$, with $\epsilon_W \in \{0.2,0.4,\ldots,0.8,2\}$. For each value of $\epsilon_W$, the FD is estimated 50 times using independently generated datasets. The average area of the resulting FDs is reported in Fig.~\ref{fig:FD_W} (solid lines), along with the minimum and maximum values across the runs (dashed lines), for the different values of $\epsilon_W$. It can be observed that, on average, DD-IMPC yields a larger feasible domain than TZPC. Moreover, TZPC exhibits a strong sensitivity to $\epsilon_W$, as the area of the FD drops significantly for larger values of the disturbance bound. Moreover, the low data sensitivity of DD-IMPC is confirmed by the small gap between the minimum and maximum area profiles, whereas this effect is much more pronounced for TZPC. 
            In addition, for $\epsilon_W \geq 0.8$, TZPC yields an empty feasible region in some runs, while for $\epsilon_W \geq 1.8$ all datasets generate an infeasible problem. 
            Once again, SM-IMPC achieves the largest feasible domain and appears to be  essentially unaffected by $\epsilon_W$. In particular Fig.~\ref{fig:FD_W} suggests that 
            DD-IMPC is competitive with SM-IMPC for moderate levels of process disturbance, while it turns out to be more conservative for higher values of $\epsilon_W$.
            \begin{figure}[h]
                \centering
                \includegraphics[width=\columnwidth]{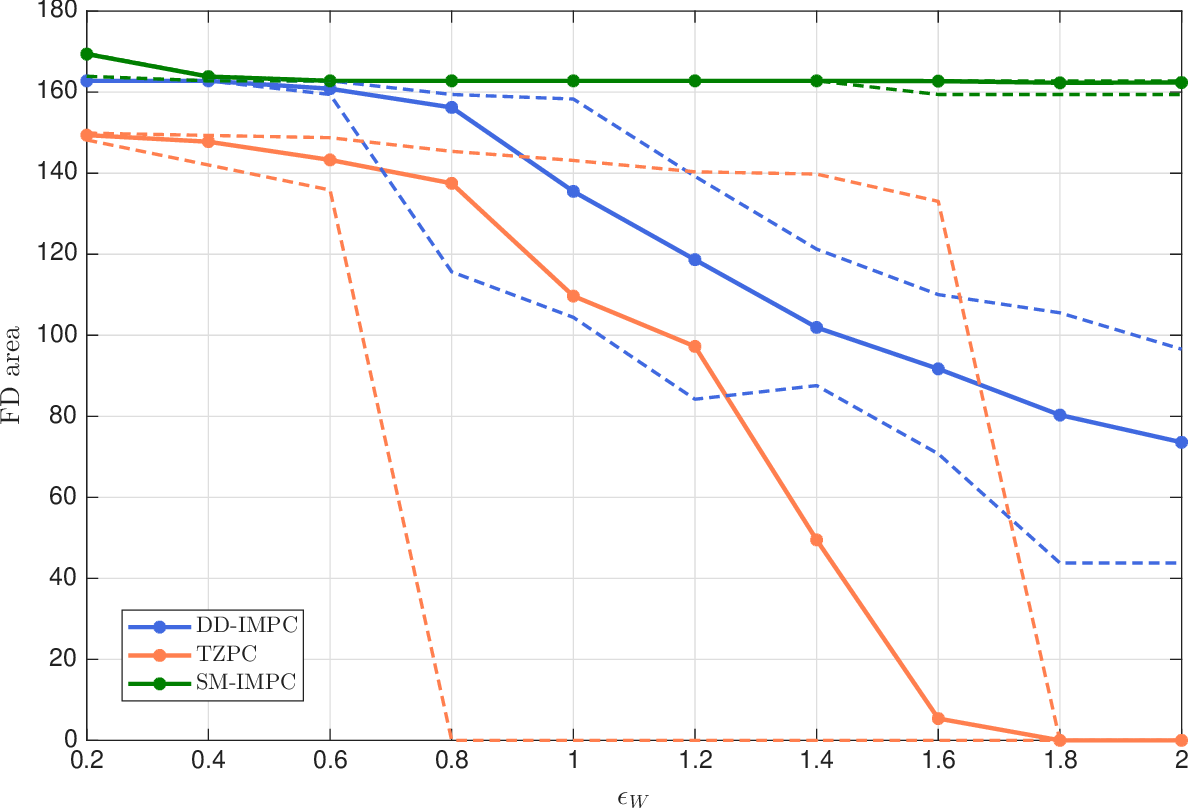}
                \caption{Area of the feasible domains as a function of $\epsilon_W$. The curves show the average (solid lines), minimum and maximum values (dashed lines) over 50 runs.}
                \label{fig:FD_W}
            \end{figure}

            Fig.~\ref{fig:MC} shows the trajectories of the closed-loop system
            resulting from one instance of the DD-IMPC scheme, for 100 different realizations of the process disturbance sequence $w(k)\in\WW$. 
            It can be observed that the constraints~\eqref{eq:constraint} are robustly satisfied and that all trajectories converge to the set $\XX_\infty$, in accordance with Theorem~\ref{thm:convergence}. 
            \begin{figure}
                \centering
                \includegraphics[width=\columnwidth]{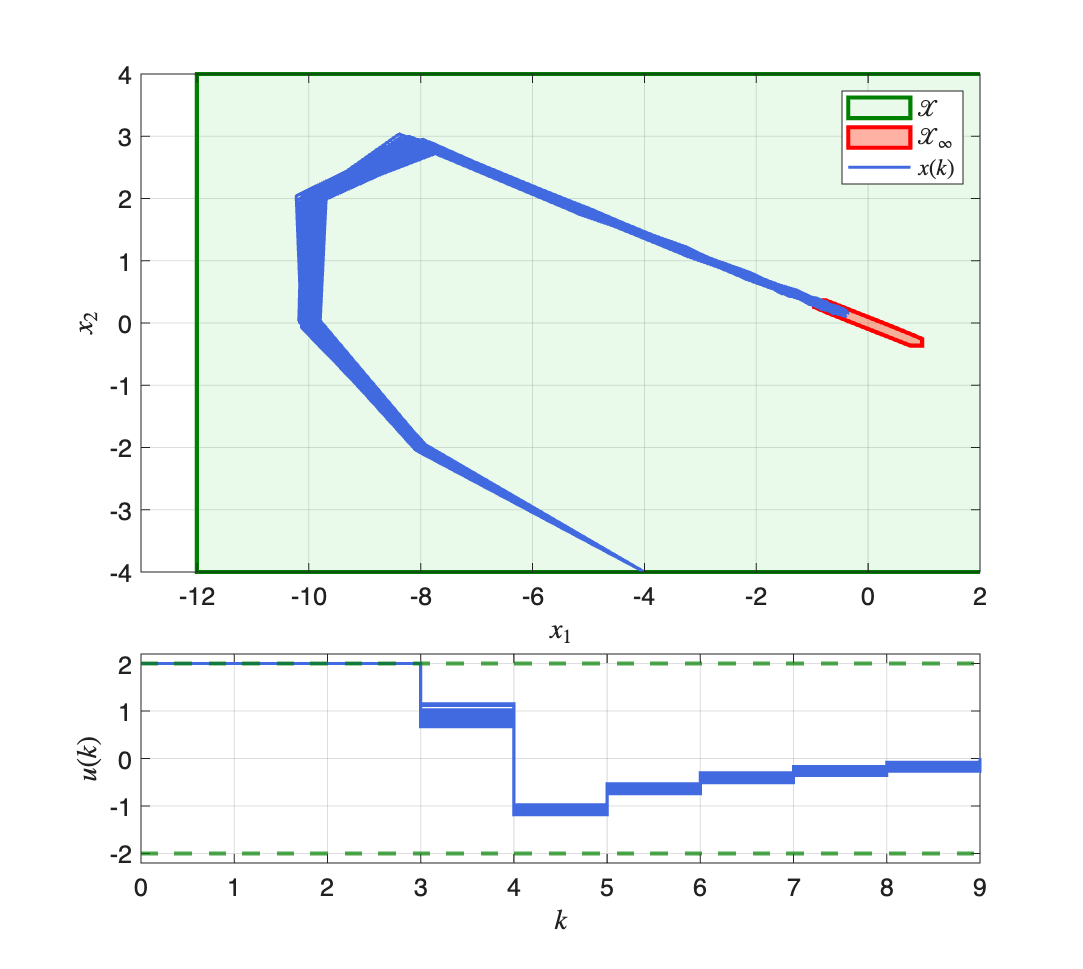}
                \caption{Closed-loop trajectories of the state (top) and input (bottom), for different realization of process disturbance sequences. The set $\XX_\infty$ is shown in red, while the state and input constraints are shown in green.}
                \label{fig:MC}
            \end{figure}
            
            In this setting, the average runtime required to solve one instance of problem~\eqref{eq:mpc_tube_box} is 9.5~ms, while that for one iteration of TZPC is 1.6~ms. This is expected, taking into account that TZPC employs a fixed horizon and a rigid tube, while \eqref{eq:mpc_tube_box} is a variable-horizon problem that explicitly accounts for uncertainty propagation along the prediction horizon.
            On the other hand, it is remarkable that the runtime of the proposed method is much smaller than that of existing DDPC techniques which incorporate uncertainty propagation within the optimal control problem, such as~\cite{alanwar2022robust, russo2023tube} that require several minutes per iteration (see \cite{farjadnia2024robust}).

	\section{Conclusions}
    \label{sec:conclusions}

A data-driven robust predictive control scheme has been presented, which is able to deal with the presence of unknown-but-bounded process disturbances affecting the unknown system dynamics. The proposed approach provides a viable alternative to the classical pipeline coupling set-membership system identification and model-predictive control. The combined use of interval matrix uncertainty representation and variable-horizon optimal control is instrumental to guarantee both recursive feasibility of the data-driven control scheme and an efficient uncertainty propagation along the prediction horizon. This translates into a computationally lightweight control scheme, compared to traditional tube-based MPC approaches.
The investigation of alternative data-driven uncertainty representations, and of the associated trade-off between conservatism and computational load of the resulting control scheme, will be the subject of future research.

    \appendix
\section{Appendix}

\subsection{Auxiliary material}

Let us first introduce some properties and lemmas that will be useful in the subsequent proofs.
The following set inclusions hold
\begin{align}
    (\mathcal{A} \oplus \mathcal{B}) \mathcal{C} \subseteq 
    \mathcal{A}\mathcal{C}  \oplus  \mathcal{B}\mathcal{C}, \label{eq:setdp1} \\
    \mathcal{A} ( \mathcal{B} \oplus \mathcal{C}) \subseteq 
    \mathcal{A}\mathcal{B}  \oplus  \mathcal{A}\mathcal{C}.
    \label{eq:setdp2} 
\end{align}
For any matrix $M$ and interval matrix $\II = C \oplus \intsimple{\Delta}$, their product satisfy the inclusions
\begin{align}
& \II M = C M \oplus \intsimple{\Delta}M \;\subseteq\; C M \oplus \intsimple{\Delta |M|}, \label{eq:intxM}    \\    
& M \II  = M C \oplus M \intsimple{\Delta} \;\subseteq\; M C \oplus \intsimple{ |M| \Delta} .\label{eq:intxMleft}        
\end{align}
The sum of two interval matrices
$\II_1 = {C_1} \oplus \intsimple{\Delta_1}$, $\II_2 = {C_2} \oplus \intsimple{\Delta_2}$ is given by
\begin{equation}
\label{eq:intsum}
\II_1 \oplus \II_2 = C_1+C_2 \oplus \intsimple{\Delta_1+\Delta_2},
\end{equation}
while the sum of two zonotopes
$\MM = \langle M;\, G_1 \dots, G_{n_g} \rangle$, 
$\mathcal{N} = \langle N;\, H_1, \dots, H_{n_h} \rangle$is equal to
\begin{equation}
\label{eq:mink_MZ}    
\MM \oplus \mathcal{N} = \langle M + N;\, 
G_1, \dots, G_{n_g},\, H_1, \dots, H_{n_h} \rangle.
\end{equation}
The multiplication between two interval matrices can be over-approximated as follows
\begin{equation}
\label{eq:intprod}
\II_1 \II_2 \subseteq {C_1} {C_2} \oplus \intsimple{ |{C_1}| \Delta_2 + \Delta_1 |{C_2}| + \Delta_1 \Delta_2 }.
\end{equation}
Let $A$ and $M$ be matrices, with $M \geq 0$. Then, the set of matrices $\mathcal{E}(M)$ satisfies
\begin{equation}
\label{eq:AdecM}
\sum_{k=1}^{lq} |A M^{(k)} | =|A| M.
\end{equation}
The following lemmas are instrumental for the proofs of the main results.
\begin{lemma}
\label{lem:Ij}
The set $\II(j)$ defined in \eqref{eq:Ibox} satisfies
\begin{equation}
\label{eq:Ijint}
\II(j)= \intsimple{ \sum_{i=0}^j | \hat{A}_K^{j-i} | F_i }
\end{equation}
where $F_0=\Delta_S$ and
\begin{equation}
\label{eq:Fj}
F_{j+1} = \Delta_K \sum_{i=0}^{j} | \hat{A}_K^{j-i} | F_i, ~~~\mbox{for}~j \geq 0. 
\end{equation}
\end{lemma}
\emph{Proof.} Let us first prove by induction that 
\begin{equation}
\label{eq:TTj}
\TT^j_{\MDK}(\II_{\Delta})=\langle 0; \hat{A}_K^j \EE(F_0),   \hat{A}_K^{j-1} \EE(F_1), \dots,  \EE(F_j) \rangle.
\end{equation}
For $j=1$, by applying \eqref{eq:Topdef} with $\II=\MDK=\hat{A}_K \oplus \intsimple{\Delta_K}$ and $\MM=\II_{\Delta} = \langle 0;\,\mathcal{E}(\Delta_S)\rangle$ one gets
$$
\TT_{\MDK}(\II_{\Delta})=\langle 0; \hat{A}_K \EE(F_0),  \EE(F_1) \rangle,
$$
with $F_0=\Delta_S$ and $F_1=\Delta_K \Delta_S$, which is in agreement with \eqref{eq:Fj}. Now, let \eqref{eq:TTj} hold for a generic $j$. By applying again \eqref{eq:Topdef},  one obtains
\begin{align*}
    &\TT^{j+1}_{\MDK}(\II_{\Delta})
    = \TT_{\MDK} \circ \TT_{\MDK}^j (\II_{\Delta}) \notag\\
    &= \big\langle 0;\,
       \hat{A}_K^{j+1} \EE(F_0),\,
       \hat{A}_K^{j} \EE(F_1),\, \dots, 
       \hat{A}_K \EE(F_{j}),\,
       \EE(F_{j+1})
    \big\rangle,
\end{align*}
where, using \eqref{eq:Fgen} and \eqref{eq:AdecM}, one has
$$
\begin{array}{rcl}
F_{j+1} &\!=\!& \displaystyle{ \Delta_K \!\!\!\!\! \sum_{i=1}^{n(n+m)}\!\!\!  \left(  |\hat{A}_K^j F_0^{(i)}|\!+ \! |\hat{A}_K^{j-1} F_1^{(i)}|\!+\! \dots \!+\! | F_j^{(i)} | \right) }\\ 
&\!=\!& \Delta_K \left( |\hat{A}_K^j| F_0 + |\hat{A}_K^{j-1} | F_1 + \dots +  F_j \right)
\end{array}
$$
which corresponds to  \eqref{eq:Fj}.
Then, \eqref{eq:Ijint} immediately follows by applying \eqref{eq:boxoutzon} to the right hand side of \eqref{eq:TTj}.   \qed
\begin{lemma}
\label{lem:Pj}
Let us define matrices $P_j$ such that $P_1=I$ and
\begin{equation}
\label{eq:Pj}
P_{j+1} = | \hat{A}_K^j | + \sum_{h=1}^{j} P_h \Delta_K | \hat{A}_K^{j-h} |, ~~~\mbox{for}~j \geq 1. 
\end{equation}
Then, matrices $F_j$ in \eqref{eq:Fj} satisfy
\begin{equation}
\label{eq:FjPj}
F_{j} = \Delta_K P_j \Delta_S, ~~~\mbox{for}~j \geq 1. 
\end{equation}
\end{lemma}
\emph{Proof.} By recursively expanding the sum in \eqref{eq:Fj}, one gets that $F_j=\Delta_K \Sigma_{j-1} \Delta_S$, where $\Sigma_j$ is the sum of the $2^j$ terms of the form $|M_1^{q_1}||M_2^{q_2}| \dots |M_h^{q_h}|$,
for $1 \leq h \leq j$,
with $M_i \in \{\hat{A}_K,\Delta_K\}$, $M_i \neq M_{i+1}$,
and $1 \leq q_i \leq j$, $\sum_{i=1}^h q_i =j$.
Similarly, by recursively expanding the sum in \eqref{eq:Pj}, one gets $P_j=\Sigma_{j-1}$, thus giving \eqref{eq:FjPj}. \qed

\subsection{Proof of Proposition \ref{prop:overapprox}}

From \eqref{eq:setdp1}, one has
\begin{equation}
\label{eq:IM}
\begin{split}
\II \MM &= (C \oplus \intsimple{\Delta}) \langle M;\, G_1, \ldots, G_g\rangle\\
&\subseteq \langle C M;\, C G_1, \ldots, C G_g \rangle \oplus \intsimple{\Delta}\MM.
\end{split}
\end{equation}
Since $\MM$ can be expressed as 
$
\MM = M \oplus \langle 0;\, G_1, \ldots, G_{n_g} \rangle,
$
by exploiting \eqref{eq:setdp2}, \eqref{eq:intxM} and \eqref{eq:boxoutzon}, one gets
\begin{equation*}
    \begin{split}
        \intsimple{\Delta}\MM & \subseteq \intsimple{\Delta}M \oplus\intsimple{\Delta}\langle0;\,G_1,\ldots,G_g\rangle\\
        & \subseteq \intsimple{\Delta|M|} \oplus \intsimple{\Delta} \intsimple{\sum_{i=1}^{n_g} |G_i|} \\
        & \overset{\eqref{eq:intprod}}{\subseteq} \intsimple{\Delta|M|} \oplus \intsimple{ \Delta \sum_{i=1}^{n_g} |G_i|}\\
        & \overset{\eqref{eq:intsum}}{=} \intsimple{\Delta\left(|M| + \sum_{i=1}^{n_g} |G_i|\right)}=\intsimple{F}\!=\! \langle 0; \EE(F) \rangle.
    \end{split}
\end{equation*}
Hence, the results follows from \eqref{eq:IM}, by using \eqref{eq:mink_MZ} and the definition of $\TT_{\II}(\MM)$ in \eqref{eq:Topdef}. \qed

\subsection{Proof of Theorem~\ref{thm:RF}}\label{app:RF}

The aim of the proof is to show that if problem~\eqref{eq:mpc_tube_box} is feasible at time $k$, then there exists a feasible solution at time $k+1$. We distinguish the cases in which the horizon length of the optimal solution at time $k$ is $N_k^* > 1$ (Case 1) or $N_k^* = 1$ (Case 2).

\textbf{Case 1 ($N_k^* >1$):}
Let us consider the following candidate solution for step $k+1$ 
\begin{equation}\label{eq:zv_candidate}
    \begin{array}{rcl}
        \hat{v}_{k+1}(j) &=&
            v^*_k(j+1) + K\AKc^j\delta(k)+ K\AKc^jw(k), \\
            && \hspace*{3cm} j = 0,\ldots,N^*_k\!-\!2 , \vspace*{1mm}\\
        \hat{z}_{k+1}(j) &=&
            z^*_k(j+1) + \AKc^j\delta(k)+\AKc^jw(k), \\
                        && \hspace*{3cm} j = 0,\ldots,N^*_k-1
    \end{array}
\end{equation}
where
\begin{equation}
\label{eq:deltak}
\delta(k) = x(k+1)-\hat{A} x(k) - \hat{B} u(k) = A_{\Delta} x(k) + B_{\Delta} u(k)
\end{equation}
and $z_k^*(j)$, $v_k^*(j)$ are the optimal nominal states and inputs for problem~\eqref{eq:mpc_tube_box}. Note that the length of this solution is $N_k^*-1$.
From \eqref{eq:deltak}, one has that
\begin{equation}
\label{eq:deltakset}
\delta(k)
 \in \Mdelta\left[\begin{array}{c}
         x(k)\\
         u(k)
    \end{array}\right] 
    = \intsimple{\Delta_S}  \myvec{z_k^*(0)}{v_k^*(0)}.
\end{equation}
Next, we show that \eqref{eq:zv_candidate} satisfies all the constraints of problem~\eqref{eq:mpc_tube_box}.

The initial state $ \hat{z}_{k+1}(0)$ of the candidate solution~\eqref{eq:zv_candidate} satisfies
\begin{equation*}
    \begin{split}
      &  \hat{z}_{k+1}(0) = z^*_k(1)+\delta(k)+w(k)\\ 
      & = \Ac z^*_k(0)+\Bc v^*_k(0)+A_{\Delta}x(k) + B_{\Delta}u(k)+w(k)\\
      &  = (\Ac + A_{\Delta})x(k) + (\Bc + B_{\Delta})u(k)+w(k) 
       = x(k+1)
    \end{split}
\end{equation*}
which corresponds to constraint \eqref{subeq:mpc_init} at time $k+1$.

By using~\eqref{eq:zv_candidate}, we have that
\begin{equation*}
    \begin{split}
        &\hat{z}_{k+1}(j+1)-\AKc^{j+1}\delta(k)- \AKc^{j+1} w(k) = z^*_k(j+2)\\ 
        & = \Ac z_k^*(j+1) + \Bc v_k^*(j+1)\\
        &= \Ac\left(\hat{z}_{k+1}(j)-\AKc^j \delta(k)  -\AKc^j w(k) \right)\\
        & \quad +\Bc\left(\hat{v}_{k+1}(j)-K\AKc^j\delta(k) - K\AKc^j w(k)\right) \\
        &= \Ac\hat{z}_{k+1}(j)\!+\!\Bc \hat{v}_{k+1}(j) - \AKc^{j+1}\delta(k) - \AKc^{j+1} w(k),
    \end{split}
\end{equation*}
which implies $\hat{z}_{k+1}(j+1) = \Ac \hat{z}_{k+1}(j)+\Bc \hat{v}_{k+1}(j)$ and hence satisfaction of constraint~\eqref{subeq:mpc_dynamic}.

To prove that the candidate solution \eqref{eq:zv_candidate} satisfies the state constraint~\eqref{subeq:mpc_state}, let us define 
\begin{align}
&
\BB_k^*(j) = \sum_{i=0}^{j-1}\left\{ \II^{\Delta}(j-i-1) \xi_k^*(i) + \II^{\WW}(i) \right\}
\label{eq:BBkstar}\\
&
\hat\BB_{k+1}(j) = \sum_{i=0}^{j-1} \left\{ \II^{\Delta}(j-i-1) \hat\xi_{k+1}(i) + \II^{\WW}(i) \right\}
\label{eq:BBkp1hat}
\end{align}
where
$$
\xi_k^*(i) = \myvec{z_k^*(i)}{v_k^*(i)}\,,~~~
\hat\xi_{k+1}(i) = \myvec{\hat{z}_{k+1}(i)}{\hat{v}_{k+1}(i)}.
$$
We aim to show that
\begin{equation}
\label{eq:setincl}
\hat{z}_{k+1}(j) \oplus \hat\BB_{k+1}(j) \subseteq z_k^*(j+1) \oplus \BB_k^*(j+1) 
\end{equation}
for $j=0,1,\dots,N_k^*-1$, which implies ~\eqref{subeq:mpc_state} because the right hand side of \eqref{eq:setincl} is included in $\XX$. 
By using the second equation in  \eqref{eq:zv_candidate}, one has that \eqref{eq:setincl} is equivalent to
\begin{equation}
\label{eq:setincl2}
\hat{A}_{K}^j (\delta(k)+w(k)) \oplus \hat\BB_{k+1}(j) \subseteq \BB_k^*(j+1). 
\end{equation}
Exploiting \eqref{eq:zv_candidate}, the left hand side of \eqref{eq:setincl2} becomes
\begin{equation*}
    \begin{aligned}
    &\hat{A}_{K}^j (\delta(k)\!+\!w(k)) \oplus \!\hat\BB_{k+1}(j)
    \\
    &= \hat{A}_{K}^j (\delta(k)\!+\!w(k)) \oplus  \sum_{i=0}^{j-1} \left\{ \II^{\Delta}(j \!-\! i \!-\! 1) \hat\xi_{k+1}(i) + \II^{\WW}(i) \right\} \\
    &= \hat{A}_{K}^j (\delta(k)\!+\!w(k)) \oplus \sum_{i=0}^{j-1} \left\{ \II^{\Delta}(j \!-\! i \!-\! 1) \xi_k^*(i \!+\! 1) + \II^{\WW}(i) \right\} \\
        &\quad
       \oplus \sum_{i=0}^{j-1} \II^{\Delta}(j \!-\! i \!-\! 1)
              M_K \hat{A}_K^i (\delta(k)\!+\!w(k))
    \end{aligned}
\end{equation*}
where $M_K=\myvec{I}{K}$.
Therefore, by writing the right hand side of \eqref{eq:setincl2} according to \eqref{eq:BBkstar}, it turns out that \eqref{eq:setincl2} holds if the following two inclusions are true
\begin{align}	
&\hat{A}_{K}^j \delta(k) \oplus \sum_{i=0}^{j-1} \II^{\Delta}(j \!-\! i \!-\! 1) M_K \hat{A}_K^i \delta(k)  \subseteq \II^{\Delta}(j) \xi_k^*(0),  \label{eq:setincl3} \\
&\hat{A}_{K}^j w(k) \oplus \sum_{i=0}^{j-1} \II^{\Delta}(j \!-\! i \!-\! 1) M_K \hat{A}_K^i w(k)  \subseteq \II^{\WW}(j).  \label{eq:setincl4} 
\end{align}
Let us first prove \eqref{eq:setincl3}. By using Lemmas \ref{lem:Ij} and \ref {lem:Pj} and recalling  $\Delta_K = \Delta_S |M_K|$, its left hand side satisfies
\begin{align}
&\hat{A}_{K}^j \delta(k)
\oplus
\sum_{i=0}^{j-1}
\II(j-i-1) M_K \hat{A}_K^i \delta(k)
\notag \\
&\overset{\eqref{eq:Ijint}}{=}
\hat{A}_{K}^j \delta(k)
\oplus\!
\sum_{i=0}^{j-1}
\intsimple{
    \sum_{h=0}^{j-1-i}
    \!|\hat{A}_K^{j-1-i-h}| F_h
}\!\!M_K \hat{A}_K^i
\delta(k)
\notag \\
&\overset{\eqref{eq:FjPj}}{=}
\hat{A}_{K}^j \delta(k)
\oplus
\sum_{i=0}^{j-1}
\intsimple{
    |\hat{A}_K^{j-1-i}| \Delta_S
}
M_K \hat{A}_K^i
\delta(k)
\notag \\
&\quad\quad
\oplus
\sum_{i=0}^{j-1}
\intsimple{
    \sum_{h=1}^{j-1-i}
    |\hat{A}_K^{j-1-i-h}|
    \Delta_K P_h \Delta_S
}
M_K \hat{A}_K^i
\delta(k)
\notag \\[2mm]
&\overset{\eqref{eq:intxM}}{\subseteq}
\hat{A}_{K}^j \delta(k)
\oplus
\sum_{i=0}^{j-1}
\intsimple{
    |\hat{A}_K^{j-1-i}|
    \Delta_K
    |\hat{A}_K^i|
}
\delta(k)
\notag \\
&\quad\quad
\oplus
\sum_{h=1}^{j-1}
\sum_{i=0}^{j-1-h}
\intsimple{
    |\hat{A}_K^{j-1-i-h}|
    \Delta_K P_h \Delta_K
    |\hat{A}_K^i|
}
\delta(k)
\label{eq:lastlhs}
\end{align}
where the last two equalities exploit property \eqref{eq:intsum} and a reorganization of the sum indexes.
Similarly, the right hand side of \eqref{eq:setincl3} satisfies
\begin{align}
& \II(j) \xi_k^*(0) =  \intsimple{ \sum_{i=0}^{j} | \hat{A}_K^{j-i} | F_i }  \xi_k^*(0) \notag \\
& \overset{\eqref{eq:FjPj}}{=}  \intsimple{| \hat{A}_K^{j} | \Delta_S + \sum_{i=1}^{j} | \hat{A}_K^{j-i} |  \Delta_K P_i \Delta_S}  \xi_k^*(0)  \notag \\
& \overset{\eqref{eq:Pj}}{=}  \intsimple{| \hat{A}_K^{j} | \Delta_S + \sum_{i=1}^{j} | \hat{A}_K^{j-i} |  \Delta_K  | \hat{A}_K^{i-1} | \Delta_S}\xi_k^*(0)\notag \\
& \quad \oplus \intsimple{\sum_{i=1}^{j} | \hat{A}_K^{j-i} | \Delta_K  \sum_{h=1}^{i-1} P_h \Delta_K | \hat{A}_K^{i-1-h} |  \Delta_S }  \xi_k^*(0)  \notag\\
& =  \intsimple{| \hat{A}_K^{j} | \Delta_S + \sum_{i=0}^{j-1} | \hat{A}_K^{j-1-i} |  \Delta_K  | \hat{A}_K^{i} | \Delta_S} \xi_k^*(0)  \notag\\
&\quad\oplus \intsimple{\sum_{i=0}^{j-1} | \hat{A}_K^{j-1-i} | \Delta_K  \sum_{h=1}^{i} P_h \Delta_K | \hat{A}_K^{i-h} |  \Delta_S }  \xi_k^*(0)  \notag \\
& = \intsimple{| \hat{A}_K^{j} | \Delta_S + \sum_{i=0}^{j-1} | \hat{A}_K^{j-1-i} |  \Delta_K  | \hat{A}_K^{i} | \Delta_S }  \xi_k^*(0)  \notag \\
&\quad\oplus\!\intsimple{ \sum_{h=1}^{j-1}  \sum_{i=0}^{j-1-h} | \hat{A}_K^{j-1-i-h} | \Delta_K  P_h \Delta_K | \hat{A}_K^{i} |  \Delta_S } \! \xi_k^*(0).   \label{eq:lastrhs}
\end{align}
Now, recall that from \eqref{eq:deltakset} one has $\delta(k) \in \intsimple{\Delta_S}\xi_k^*(0)$. Therefore, by comparing \eqref{eq:lastlhs} and \eqref{eq:lastrhs}, and using properties \eqref{eq:intprod} and \eqref{eq:intxMleft}, one gets that \eqref{eq:setincl3} holds.
Inclusion \eqref{eq:setincl4} can be proved in the same way, exploiting $w(k) \in \intsimple{\overline{W}}$.
For the input constraint~\eqref{subeq:mpc_input}, a similar reasoning can be applied.

Finally, inclusion \eqref{eq:setincl} allows one to prove also the terminal constraint~\eqref{subeq:mpc_terminal}. In fact, one has
$$
\hat{z}_{k+1}(N_k^*-1) \oplus \hat{\BB}_{k+1}(N_k^*-1) \!\subseteq \!z_k^*(N_k^*)\oplus \BB^*_k(N_k^*) \subseteq \XX_f.
$$

\textbf{Case 2 ($N_k^* = 1$):}
For time instant $k+1$, let us consider the candidate solution
\begin{equation}\label{eq:candidate2}
\begin{array}{l}
     \hat{v}_{k+1}(0) = Kx(k+1),  \\
     \hat{\textbf{z}}_{k+1} = \left(x(k+1),\,\AKc x(k+1)\right). 
\end{array}
\end{equation}
with associated horizon length $\hat{N}_{k+1} = N_k^* =1$.
Initial and dynamic constraints~\eqref{subeq:mpc_init}-\eqref{subeq:mpc_dynamic} are trivially satisfied by the candidate solution~\eqref{eq:candidate2}.
Now, let us show that $x(k+1)\in \XX_f$. From optimality of problem \eqref{eq:mpc_tube_box} (whose optimal horizon length is $N_k^*=1$), we know that $z_k^*(1) \oplus \BB_k^*(1) \subseteq \XX_f$. 
Therefore
\begin{equation}
\label{eq:xkinXf}
    \begin{split}
        & x(k+1) = A x(k) + B u(k) +w(k) \\
        & = \Ac x(k) + \Bc u(k) + \delta(k) +w(k) \\
        & = z^*_{k}(1) + \delta(k) + w(k) \in z^*_{k}(1) \oplus \Mdelta\myvec{z^*_k(0)}{v^*_k(0)} + \WW \\
               & = z_k^*(1) \oplus \BB_k^*(1) \subseteq \XX_f.
    \end{split}
\end{equation}
Hence, the state constraint \eqref{subeq:mpc_state} holds due to $\XX_f \subseteq \XX$.
Since $\BB_{k+1}(0) = \{0\}$ and $ K \XX_f \subseteq \UU$, one has also
$
\hat{v}_{k+1}(0) = Kx(k+1) \in K\XX_f \subseteq \UU,
$
which corresponds to the input constraint~\eqref{subeq:mpc_input}. For the terminal constraint, one has
\begin{equation*}
\begin{split}
&    \hat{z}_{k+1}(1)\oplus \hat{\BB}_{k+1}(1)  = \!\AKc x(k+1) \oplus \Mdelta \! \myvec{x(k+1)}{\!\!Kx(k+1)\!\!} \!\oplus\! \WW \\
    & = \!\AKc x(k+1) \oplus \Mdelta \myvec{I}{K}\!x(k+1) \oplus \WW \\
    & = \left( \AKc \oplus \Mdelta\myvec{I}{K}\right)x(k+1) \oplus \WW\\
    & = \AK x(k+1) \oplus \WW \subseteq \XX_f,
\end{split}
\end{equation*}
thanks to Assumption \ref{assum:RPI}, being $x(k+1)\in \XX_f$. \qed

\subsection{Proof of Theorem~\ref{thm:convergence}}\label{app:convergence}

The cost $\hat{J}_{k+1}$ associated to the candidate solution~\eqref{eq:zv_candidate} satisfies 
\begin{equation*}
    \begin{split}
        \hat{J}_{k+1} & = \gamma\hat{N}_{k+1}+\sum_{j=0}^{\hat{N}_{k+1}-1} \ell\left(\hat{v}_{k+1}(j)-K\hat{z}_{k+1}(j)\right)\\
        & = \gamma (N_k^*-1) + \sum_{j=0}^{N^*_k-2}\! \ell\left({v}^*_{k}(j+1)-K{z}^*_{k}(j+1)\right)\\
        & =\gamma N_k^*-\gamma + \sum_{j=1}^{N^*_k-1} \ell\left({v}^*_{k}(j)-K{z}^*_{k}(j)\right)\\
        & = J^*_k -\gamma - \ell\left({v}^*_{k}(0)-K{z}^*_{k}(0)\right)\leq J^*_k - \gamma.
    \end{split}
\end{equation*}
Therefore, the optimal cost at time $k+1$ is such that
$J_{k+1}^* \leq \hat{J}_{k+1}\leq J_k^* - \gamma$.
This ensures that there exists a time $\bar{k} \leq \lfloor J_0^*/\gamma\rfloor$ at which the optimal horizon length is $N_{\bar{k}}^*=1$, and hence $x(\bar{k}+1) \in \XX_f$ according to~\eqref{eq:xkinXf}.

When $N_k^* =1$, the cost associated with the candidate solution~\eqref{eq:candidate2} is $\hat{J}_{k+1} = \gamma$, which implies that such a solution is also optimal (as the cost value cannot be less than $\gamma$). Consequently, for all $k > \bar{k}$ one has $J^*_k = \gamma$ and the optimal solution of problem~\eqref{eq:mpc_tube_box} is the candidate solution~\eqref{eq:candidate2}.
Then, due to Assumption \ref{assum:RPI}, the trajectories of the closed-loop system will remain indefinitely inside the terminal set $\XX_f$. 
Moreover, since for $k > \bar{k}$ one has $x(k+1) = A_K x(k) +w(k)$, all trajectories will converge to the set $\sum_{i=0}^{+\infty} \AK^i \WW \subseteq \sum_{i=0}^{+\infty} \TT^i_{\MDK}(\WW) =\XX_\infty$. \qed

	\bibliographystyle{ieeetr}
	\bibliography{biblioVH.bib}

\end{document}